\documentclass[11pt]{article}
\usepackage{defs}

\title{A Tight Cycle-Cover Inequality for Shortest Common Superstring}
\author{
   Nikolai Chukhin
   \thanks{JetBrains Research.
   Email: \url{buyolitsez1951@gmail.com}}
   \and
   Alexander S. Kulikov
   \thanks{JetBrains Research.
   Email: \url{alexander.s.kulikov@gmail.com}}
   \and
   Ivan Mihajlin
   \thanks{JetBrains Research.
   Email: \url{ivmihajlin@gmail.com}}
   \and
   Alexander Smal
   \thanks{JetBrains Research.
   Email: \url{avsmal@gmail.com}}
}
\date{}

\begin{document}
\maketitle
\begin{abstract}
	In~the Shortest Common Superstring problem (SCS), one is~given a~finite set of~strings and is~asked to~find a~shortest string containing every input string as~a~substring.
	Its best known approximation ratio is~$2.466$, whereas
	the currently strongest upper bound on~the approximation guarantee 
	of~the maximum-overlap greedy algorithm is~$3.396$ 
	(Englert, Matsakis, and Vesel{\'y}, 2023),
	though it~is conjectured to~be~2.
	We~improve both approximation guarantees: 
	SCS admits a~$\frac{7}{3}$-approximation
	and 
	the approximation guarantee of~the greedy algorithm is~at~most~$3$.
	
	The main technical ingredient of~our proof 
	is~a~certain inequality for optimum cycle covers 
	of~an~overlap graph associated with the input strings.
	Every previous improvement of~greedy's worst-case guarantee and the two recent record guarantees for general SCS are~driven by~it.
	We~improve this inequality by~pushing~it to~its limit:
	for a~particular coefficient of~this inequality,
	we~show a~new upper bound and prove that 
	it~cannot be~improved further.
\end{abstract}

\section{Approximating Shortest Common Superstring}
\label{sec:introduction}

In~the Shortest Common Superstring problem (SCS), one is~given a~finite set of~strings and is~asked to~find a~shortest string containing every input string as~a~substring.
The problem is~a~classical model of~assembling highly overlapping data and appears in~data compression and genome assembly~\cite{GP14}.
It~is~\NP-hard already for strings of~length three~\cite{GMS80}, 
and, for this reason, polynomial-time approximation algorithms are actively studied.
The first constant-factor algorithm, with approximation ratio~$3$, was given in~\cite{BJLTY91}.
Subsequent work improved this ratio through increasingly delicate combinations of~cycle covers and maximum asymmetric traveling salesperson algorithms.
The currently best known approximation ratio is~$2.466$~\cite{EMV23}.
Interestingly, there~is a~simple greedy algorithm whose conjectured \cite{TU88,Storer88,Turner89,BJLTY91} 
approximation ratio is~equal to~2:
while more than one string remains, choose an~ordered pair with maximum suffix--prefix overlap, merge the pair, and iterate.
This algorithm admits a~linear-time implementation over a~constant-size alphabet~\cite{Ukkonen90} and is~known to~behave well in~practice~\cite{RBT04,Ma09,SVB23}.
Whereas it~is not difficult to~show that its approximation ratio is~at~least~2,
the best known upper bound is~$3.396$~\cite{EMV23}.

\subsection*{Our Contribution}
We~improve both approximation guarantees: 
SCS admits a~$\frac{7}{3}$-approximation
and 
the approximation guarantee of~the greedy algorithm is~at~most~$3$.
In~\Cref{figure:timeline}, we~show a~timeline
of~all known upper bounds on~the approximation ratio of~SCS
and the greedy algorithm.

\begin{figure}
	\begin{center}
		\begin{tikzpicture}[xscale=.39, yscale=2.2]
			\tikzstyle{l} = [gray!20, thin]
			\tikzstyle{d} = [circle, inner sep=0mm, minimum size=.7mm, fill=black]
			
			\foreach \year in {1990, ..., 2027}
			\draw[l] (\year - 1990, 2) -- (\year - 1990, 4);
			\foreach \year in {1995, 2000, ..., 2025}
			\node[above] at (\year-1990, 4) {\year};
			\foreach \ratio in {4, 3, 2} {
				\draw[l] (0, \ratio) -- (37, \ratio);
				\node[left] at (0, \ratio) {\ratio};
			}
			
			\foreach \year/\ratio/\ref [count=\n] in {
				1991/3.000/BJLTY91,
				1993/2.889/TY93,
				1994/2.75/{CGPR94, KPS94, AS94},
				1995/2.725/AS95,
				1996/2.667/AS96,
				1997/2.596/BJJ97,
				1999/2.5/Sweedyk99,
				2003/2.5/KLSS03,
				2005/2.5/KS05,
				2012/2.5/PEZ12,
				2013/2.479/Mucha13,
				2022/2.475/EMV22,
				2023/2.466/EMV23} {
				\node[d] (\n) at (\year-1990, \ratio) {};
				\node[below] at (\year-1990, 2) {\small \rotatebox{90}{\cite{\ref}}};
			}
			
			\node[d] at (1994-1990, 2.794) {};
			\node[d] at (1994-1990, 2.834) {};
			
			\foreach \year/\ratio/\ref [count=\n] in {
				1991/4.000/BJLTY91,
				2005/3.5/KS05,
				2022/3.425/EMV22,
				2023/3.396/EMV23} {
				\node[d] (g\n) at (\year-1990, \ratio) {};
			}

			\node[d] (14) at (2026-1990, 2.333) {};
			\node[d] (g5) at (2026-1990, 3) {};
			\node[below] at (2026-1990, 2) {\small \rotatebox{90}{this paper}};
			
			\foreach \n in {1,...,13} {
				\tikzmath{\m=int(\n+1);}
				\draw (\n) -| (\m);
			}
			\foreach \n in {1,...,4} {
				\tikzmath{\m=int(\n+1);}
				\draw (g\n) -| (g\m);
			}
			
			\node[right] (greedy) at (2027-1990+.2, 3) {greedy}; 
			\draw (g5) -- (greedy);
			\node[right] (scs) at (2027-1990+.2, 2.333) {SCS};
			\draw (14) -- (scs);
		\end{tikzpicture}
		\vspace{-7mm}
	\end{center}
	\caption{Timeline of~published approximation ratios for SCS (bottom curve) and 
		the greedy algorithm for SCS (top curve).}
	\label{figure:timeline}
\end{figure}
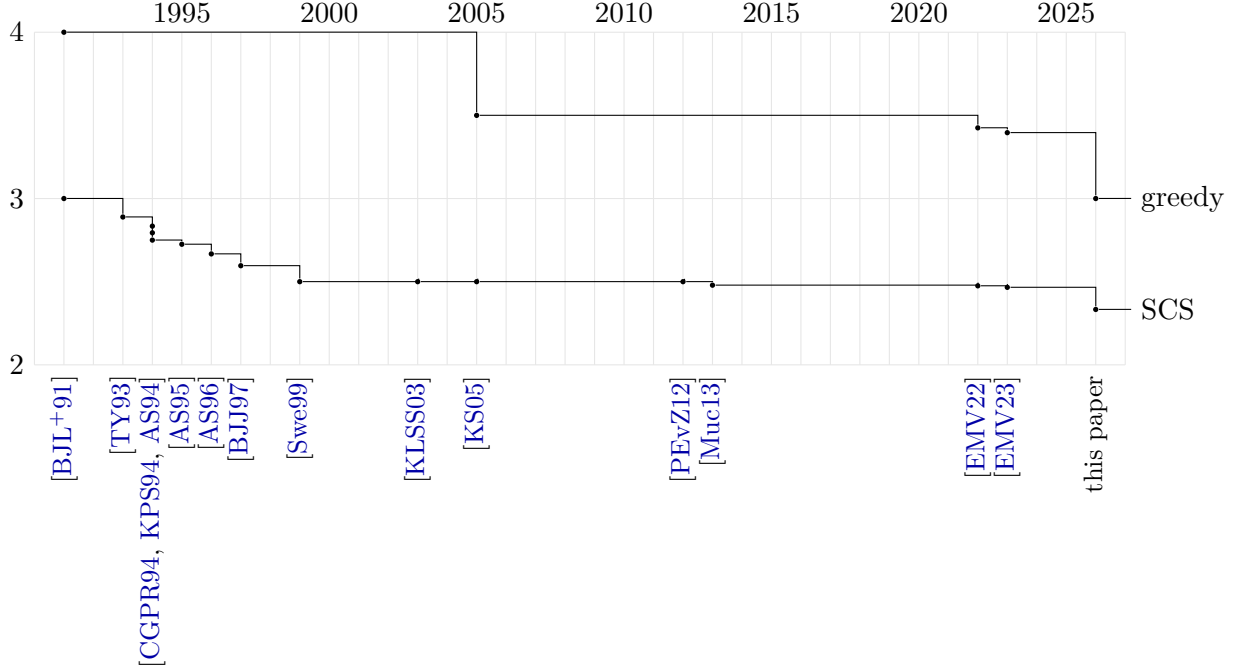



\section{Known Technical Results}
We~start by~listing important known results on~SCS; we introduce
related concepts along the way.
Gray boxes are used to~visualize these concepts and facts; 
they are intended to~help the reader and may~be safely skipped.

In~the rest of~the paper, $\mathcal{S}=\{s_1, \dotsc, s_n\}$ is
an~input to~SCS; we~assume that no~string in~$\mathcal{S}$
is a~substring of~another string from~$\mathcal{S}$
(one can get rid of~such cases on~the preprocessing stage).
By~$\operatorname{OPT}$ we~denote an~optimal common 
superstring of~$\mathcal{S}$ or~its length. 

\subsection{Prefixes and Overlaps}
For nonempty strings~$s$ and~$t$, 
their \emph{overlap} $\operatorname{ov}(s, t)$ 
is~the longest (possibly empty) string~$q$ such that $s=p \circ q$ and $t=q \circ r$,
for some \emph{nonempty} strings $p$~and~$r$ that
are in~turn denoted by~$\operatorname{pref}(s,t)$ and 
$\operatorname{suff}(s, t)$. For example,
$\operatorname{ov}({\tt ababa}, {\tt babc})={\tt ba}$
and 
$\operatorname{ov}({\tt ababa}, {\tt ababa})={\tt aba}$.
By~$\langle s, t \rangle$ we~denote the \emph{overlap merge} of~$s$ and~$t$:
\[
	\langle s, t \rangle=
	s \circ \operatorname{suff}(s, t)=
	\operatorname{pref}(s, t) \circ t = 
	\operatorname{pref}(s, t) \circ \operatorname{ov}(s, t) \circ \operatorname{suff}(s, t).
\]
Although the~overlap and the~prefix are defined as strings, we use these two words to denote the~lengths of the~corresponding strings.
The notion of~the overlap merge extends to~more than two strings 
in~a~natural way: 
for $t_1, \dotsc, t_k$, 
\[
	\langle t_1, \dotsc, t_k \rangle=
	\operatorname{pref}(t_1, t_2) 
	\circ 
	\operatorname{pref}(t_2, t_3)
	\circ
	\dotsb
	\circ
	\operatorname{pref}(t_{k-1}, t_k)
	\circ
	t_k.
\]
Clearly, 
\[
	|\langle t_1, \dotsc, t_k \rangle| = 
	\sum_{i \in [k-1]}|\!\operatorname{pref}(t_i, t_{i+1})|+|t_k|=
	\sum_{i \in [k]}|t_i| - \sum_{i \in [k-1]}|\!\operatorname{ov}(t_i, t_{i+1})|.
\]
Hence, the smaller the sum of~adjacent prefixes, the larger
the sum of~adjacent overlaps.

%
%
%

\subsection{Permutations and Related Graphs}
The goal in~the SCS problem is~to find a~permutation of~$\mathcal{S}$ 
whose overlap merge length is~minimum. This~is equivalent to~maximizing 
the total overlap or~minimizing the total prefix (up~to
the length of~the last string in the permutation).
The greedy algorithm tries to~construct a~permutation with large total overlap by~taking, at~every step, the largest available overlap.

\begin{graybox}
	For a~set of~strings 
	\(\{{\tt abcdab}, {\tt cdabc}, {\tt dabcd}, {\tt ababa}, {\tt aaa}, {\tt babab}\}\),
	a~naive superstring results from concatenating them and has length  
	$6+5+5+5+3+5=29$:
	\({\tt abcdabcdabcdabcdababaaaababab}\).
	An~optimal superstring has length~$14$:
	\texttt{aaabababcdabcd}. 
	It~corresponds to~a~permutation 
	\((\texttt{aaa}
	,
	\texttt{ababa}
	,
	\texttt{babab}
	,
	\texttt{abcdab}
	,
	\texttt{cdabc}
	,
	\texttt{dabcd})\).
	Its total overlap is~$1+4+2+4+4=15$, exactly the number of~symbols saved relative to~the naive concatenation of~length~$29$.

	The largest overlap is~$4$; therefore, the greedy algorithm may start by~merging $\langle \texttt{dabcd}, \texttt{abcdab}, \texttt{cdabc} \rangle$, after which it~may merge $\langle \texttt{babab}, \texttt{ababa}\rangle$.
	Afterwards, there are no~overlaps larger than~$1$, and the algorithm may finish with the following superstring of~length~$16$:
	\[
		\langle \texttt{babab}, \texttt{ababa}, \texttt{aaa}, \texttt{dabcd}, \texttt{abcdab}, \texttt{cdabc}\rangle = \texttt{bababaaadabcdabc}.
	\]
\end{graybox}

All pairwise overlaps between input strings can~be conveniently represented 
in~an~\emph{overlap graph}:
it~is a~complete directed
graph with vertices~$\mathcal{S}$ and edges $\mathcal{S} \times \mathcal{S}$ (thus, self-loops are included) where the length of~an~edge $(s, t)$ is $|\!\operatorname{ov}(s, t)|$. A~\emph{prefix graph} 
is~defined similarly with the only difference that the lengths are $|\!\operatorname{pref}(s, t)|$. 

Solving SCS is~the same as~solving MAX-ATSP in~the overlap graph, that~is,
finding a~Hamiltonian path of~maximum total length. At~the same time, approximation algorithms for MAX-ATSP do~not give approximation guarantees for SCS directly since the total overlap may be~much larger
than $\operatorname{OPT}$. 

\tikzstyle{v}=[rectangle, inner sep=.5mm, draw]
\tikzstyle{l}=[pos=.19, rectangle, inner sep=.2mm, sloped, fill=fc]
\tikzstyle{e}=[bend right=10, ->, >=latex, thin]
\tikzstyle{le}=[->, >=latex, thin]

\begin{graybox}
	The overlap and prefix graphs for our running example are~shown below. In~the overlap graph, a~Hamiltonian path of~maximum
	total overlap~15 is~highlighted.
	
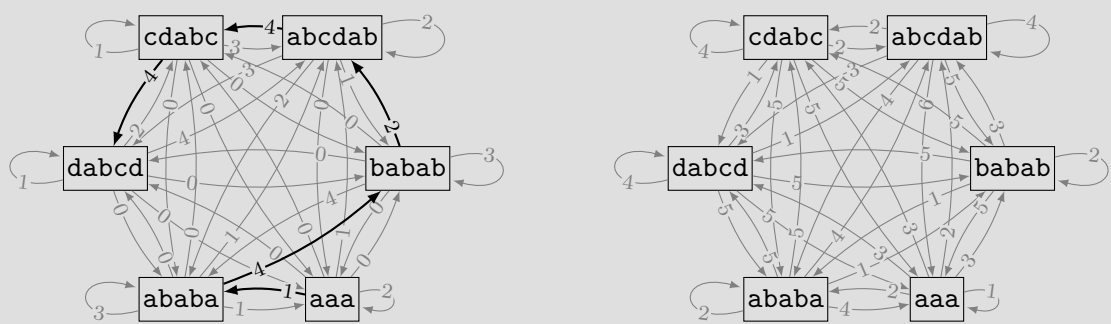
\begin{center}
		\begin{tikzpicture}
			\begin{scope}
				\foreach \s [count=\n] in {abcdab, cdabc, dabcd, ababa, aaa, babab}
					\node[v] (\s) at (60 * \n:20mm) {\strut \tt \s};
					
				\path (abcdab) edge[le, loop right, gray] node[l] {\scriptsize 2} (abcdab);
				\path (cdabc) edge[le, loop left, gray] node[l] {\scriptsize 1} (cdabc);
				\path (dabcd) edge[le, loop left, gray] node[l] {\scriptsize 1} (dabcd);
				\path (ababa) edge[le, loop left, gray] node[l] {\scriptsize 3} (ababa);
				\path (aaa) edge[le, loop right, gray] node[l] {\scriptsize 2} (aaa);
				\path (babab) edge[le, loop right, gray] node[l] {\scriptsize 3} (babab);
				
				\foreach \f/\t/\p/\o in {abcdab/cdabc/2/4,
					abcdab/dabcd/3/3,
					abcdab/ababa/4/2,
					abcdab/aaa/6/0,
					abcdab/babab/5/1,
					cdabc/abcdab/2/3,
					cdabc/dabcd/1/4,
					cdabc/ababa/5/0,
					cdabc/aaa/5/0,
					cdabc/babab/5/0,
					dabcd/abcdab/1/4,
					dabcd/cdabc/3/2,
					dabcd/ababa/5/0,
					dabcd/aaa/5/0,
					dabcd/babab/5/0,
					ababa/abcdab/4/1,
					ababa/cdabc/5/0,
					ababa/dabcd/5/0,
					ababa/aaa/4/1,
					ababa/babab/1/4,
					aaa/abcdab/2/1,
					aaa/cdabc/3/0,
					aaa/dabcd/3/0,
					aaa/ababa/2/1,
					aaa/babab/3/0,
					babab/abcdab/3/2,
					babab/cdabc/5/0,
					babab/dabcd/5/0,
					babab/ababa/1/4,
					babab/aaa/5/0}
				\path (\f) edge[e, gray] node[l] {\scriptsize \o} (\t);
				
				\foreach \f/\t/\o in {aaa/ababa/1, ababa/babab/4, babab/abcdab/2, abcdab/cdabc/4, cdabc/dabcd/4}
					\path (\f) edge[e, color=black, thick] node[l, color=black, fill=gray!25] {\scriptsize \o} (\t);
			\end{scope}
			
			\begin{scope}[xshift=80mm]
				\foreach \s [count=\n] in {abcdab, cdabc, dabcd, ababa, aaa, babab}
					\node[v] (\s) at (60 * \n:20mm) {\strut \tt \s};
					
				\path (abcdab) edge[le, loop right, gray] node[l] {\scriptsize 4} (abcdab);
				\path (cdabc) edge[le, loop left, gray] node[l] {\scriptsize 4} (cdabc);
				\path (dabcd) edge[le, loop left, gray] node[l] {\scriptsize 4} (dabcd);
				\path (ababa) edge[le, loop left, gray] node[l] {\scriptsize 2} (ababa);
				\path (aaa) edge[le, loop right, gray] node[l] {\scriptsize 1} (aaa);
				\path (babab) edge[le, loop right, gray] node[l] {\scriptsize 2} (babab);
					
					\foreach \f/\t/\p/\o in {abcdab/cdabc/2/4,
						abcdab/dabcd/3/3,
						abcdab/ababa/4/2,
						abcdab/aaa/6/0,
						abcdab/babab/5/1,
						cdabc/abcdab/2/3,
						cdabc/dabcd/1/4,
						cdabc/ababa/5/0,
						cdabc/aaa/5/0,
						cdabc/babab/5/0,
						dabcd/abcdab/1/4,
						dabcd/cdabc/3/2,
						dabcd/ababa/5/0,
						dabcd/aaa/5/0,
						dabcd/babab/5/0,
						ababa/abcdab/4/1,
						ababa/cdabc/5/0,
						ababa/dabcd/5/0,
						ababa/aaa/4/1,
						ababa/babab/1/4,
						aaa/abcdab/2/1,
						aaa/cdabc/3/0,
						aaa/dabcd/3/0,
						aaa/ababa/2/1,
						aaa/babab/3/0,
						babab/abcdab/3/2,
						babab/cdabc/5/0,
						babab/dabcd/5/0,
						babab/ababa/1/4,
						babab/aaa/5/0}
					\path (\f) edge[e, gray] node[l] {\scriptsize \p} (\t);
					
			\end{scope}
		\end{tikzpicture}
	\end{center}
\end{graybox}

The greedy superstring algorithm can be~conveniently represented in~terms
of~the overlap graph: it~constructs a~Hamiltonian path by~scanning 
the list~of~edges in~the order of~nonincreasing length and takes the current edge into
a~solution if~it does not create a~node of~in-degree or~out-degree more than one and 
does not close a~cycle. This way, at~any point of~the execution of~the greedy algorithm,
its solution is a~collection of~disjoint paths.

\begin{graybox}
	The greedy algorithm starts by~considering edges of~length~$4$.
	It~may take edges ${\tt dabcd} \to {\tt abcdab}$ and ${\tt abcdab} \to {\tt cdabc}$.
	Then, it~skips the edge ${\tt cdabc} \to {\tt dabcd}$ as~it~would create a~cycle.
	The greedy algorithm~then takes an~edge ${\tt babab} \to {\tt ababa}$. Finally, by~skipping a~bunch of~edges,
	it~takes edges ${\tt ababa} \to {\tt aaa}$ and ${\tt aaa} \to {\tt dabcd}$.
	
	
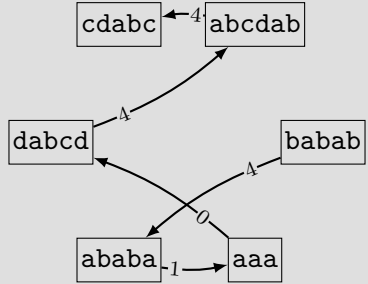
\begin{center}
		\begin{tikzpicture}
			\foreach \s [count=\n] in {abcdab, cdabc, dabcd, ababa, aaa, babab}
				\node[v] (\s) at (60 * \n:18mm) {\strut \tt \s};
				
			\foreach \f/\t/\o in {dabcd/abcdab/4, abcdab/cdabc/4, babab/ababa/4, ababa/aaa/1, aaa/dabcd/0}
				\path (\f) edge[e, color=black, thick] node[l, color=black, fill=gray!25] {\scriptsize \o} (\t);
		\end{tikzpicture}
	\end{center}
\end{graybox}

\subsection{Cycle Covers}
If~one allows the greedy algorithm to~produce cycles
(without violating the in-degree and out-degree constraints), then 
it~constructs an~optimal \emph{cycle cover}~$\mathcal{C}$, that~is, a~collection of~cycles
that cover every node of~the graph and has maximum total overlap.
This is~proved in~\cite[Theorem~10]{BJLTY91}.
The constructed cycle cover may contain self-loops.
It~is not difficult to~see that $\mathcal{C}$~is also a~cycle cover
of~minimum total length in~the prefix graph.
Indeed, consider a~cycle $C=(t_0, t_1, \dotsc, t_{r-1}, t_0)$ of~$\mathcal{C}$.
Its total overlap is~equal
to~$\sum_{i \in [r]} |\!\operatorname{ov}(t_{i-1}, t_{i})|$
and its total prefix
is~$\sum_{i \in [r]} |\!\operatorname{pref}(t_{i-1}, t_{i})|$ 
(here, indices are taken modulo~$r$). 
Note, however, that 
$|\!\operatorname{pref}(t_{i-1}, t_{i})|+|\!\operatorname{ov}(t_{i-1}, t_{i})|=|t_{i-1}|$.
Consequently, the total overlap of~$C$ and its total prefix sum to~the total length of~the strings on~$C$.
Thus, maximizing total overlap minimizes total prefix.

For $C \in \mathcal{C}$, let $w(C)$ be~the total prefix of~the cycle~$C$ and $W=\sum_{C \in \mathcal{C}}w(C)$~be the total prefix of~$\mathcal{C}$. Then,
\begin{equation}
	W \le \operatorname{OPT}.
	\label{eq:wopt}
\end{equation}
Indeed, let $(t_1, \dotsc, t_n)$ be~a~permutation of~the input strings whose overlap merge has length~$\operatorname{OPT}$. Then, the total prefix of~the cycle
$(t_1, \dotsc, t_n, t_1)$ is~at~most $\operatorname{OPT}$, but at~least~$W$
(since this single cycle is also a~cycle cover):
\[\operatorname{OPT}=\sum_{i \in [n-1]}|\!\operatorname{pref}(t_i, t_{i+1})|+|t_n| \ge \sum_{i \in [n-1]}|\!\operatorname{pref}(t_i, t_{i+1})|+|\!\operatorname{pref}(t_n, t_1)| \ge W.\]

\begin{graybox}
	The figure below shows the same optimal cycle cover 
	in~the overlap and prefix graphs.
	Its total overlap and total prefix are~$22$ and~$W = 7$, respectively; their sum is~$29$, the total length of~the input strings.
	
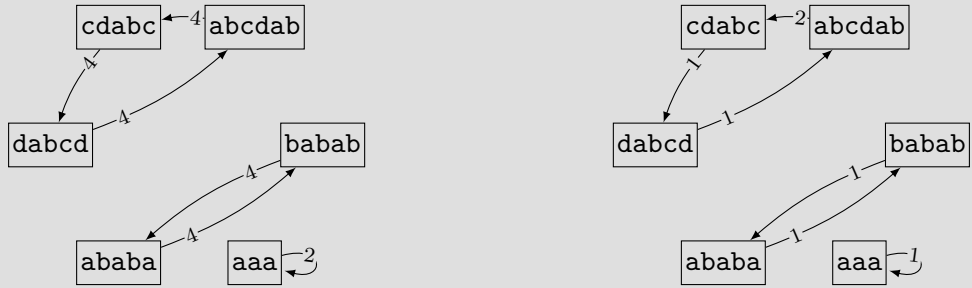
\begin{center}
		\begin{tikzpicture}
			\begin{scope}
				\foreach \s [count=\n] in {abcdab, cdabc, dabcd, ababa, aaa, babab}
				\node[v] (\s) at (60 * \n:18mm) {\strut \tt \s};
				
				\foreach \f/\t/\o in {abcdab/cdabc/4, cdabc/dabcd/4, dabcd/abcdab/4, babab/ababa/4, ababa/babab/4}
				\path (\f) edge[e, black] node[l, fill=fc] {\scriptsize \o} (\t);
				
				\path (aaa) edge[le, loop right] node[l] {\scriptsize 2} (aaa);
			\end{scope}
			
			\begin{scope}[xshift=80mm]
				\foreach \s [count=\n] in {abcdab, cdabc, dabcd, ababa, aaa, babab}
				\node[v] (\s) at (60 * \n:18mm) {\strut \tt \s};
				
				\foreach \f/\t/\o in {abcdab/cdabc/2, cdabc/dabcd/1, dabcd/abcdab/1, babab/ababa/1, ababa/babab/1}
				\path (\f) edge[e] node[l] {\scriptsize \o} (\t);
				
				\path (aaa) edge[le, loop right] node[l] {\scriptsize 1} (aaa);
			\end{scope}
		\end{tikzpicture}
	\end{center}
\end{graybox}

\subsection{Opening the Cycles}
To~turn a~cycle cover into a~Hamiltonian path, one may proceed as~follows: for each cycle, remove its edge of~minimum overlap; then, merge the resulting paths into a~Hamiltonian path.

Formally, 
consider a~cycle $C=(t_1, \dotsc, t_{r}, t_1)$ of~$\mathcal{C}$
and assume that the edge $(t_{r}, t_1)$ was the last edge of~this cycle added by~the algorithm.
We~call~it a~\emph{cycle-closing edge}; it~has minimum overlap among the edges of~the cycle.
Denote its length by $o(C)$ and let $O=\sum_{C \in \mathcal{C}}o(C)$. Now, let
\(\sigma(C)=\langle t_1, \dotsc, t_{r}\rangle\) be~a~\emph{representative string} for~$C$.
Then,
\[
	\sigma(C)=
	\operatorname{pref}(t_1, t_2) 
	\circ 
	\dotsb 
	\circ
	\operatorname{pref}(t_{r-1}, t_r)
	\circ
	t_r 
	=
	\operatorname{pref}(t_1, t_2) 
	\circ 
	\dotsb 
	\circ
	\operatorname{pref}(t_{r-1}, t_r)
	\circ
	\operatorname{pref}(t_{r}, t_1)
	\circ
	\operatorname{ov}(t_{r}, t_1).
\]
Hence,
\[|\sigma(C)|=w(C)+o(C).\]
An~algorithm that constructs an~optimal cycle cover and returns a~concatenation is~called MGREEDY; it~was introduced in~\cite{BJLTY91}.
The resulting superstring has length
\begin{equation}
	\sum_{C \in \mathcal{C}}(w(C)+o(C))=W+O.
	\label{eq:wo}
\end{equation}

\begin{graybox}
	The cycle-closing edges of~the optimal cycle cover are shown as~dashed edges below. 
	Their total overlap is~$O=4+4+2=10$.
	To~the right, we~show the corresponding representative strings.
	By~concatenating the representative strings, one gets a~superstring 
	{\tt abcdabcdbababaaaa} 
	of~length $W+O=7+10=17$. 
	
	\begin{center}
		\begin{tikzpicture}
				\foreach \s [count=\n] in {abcdab, cdabc, dabcd, ababa, aaa, babab}
					\node[v] (\s) at (60 * \n:18mm) {\strut \tt \s};
				
				\foreach \f/\t/\o in {abcdab/cdabc/4, cdabc/dabcd/4, babab/ababa/4}
					\path (\f) edge[e, black] node[l, fill=fc] {\scriptsize \o} (\t);
					
				\foreach \f/\t/\o in {dabcd/abcdab/4, ababa/babab/4}
					\path (\f) edge[e, black, dashed] node[l, fill=fc] {\scriptsize \o} (\t);
				
				\path (aaa) edge[le, loop right, dashed] node[l] {\scriptsize 2} (aaa);

				\node[right, text width=60mm] at (3, 0) {
					\begin{align*}
						\sigma(C_1) &= {\tt abcdabcd}\\
						\sigma(C_2) &= {\tt bababa}\\
						\sigma(C_3) &= {\tt aaa}\\
					\end{align*}
				};
		\end{tikzpicture}
	\end{center}
\end{graybox}

\subsection{Upper Bounding the Overlap of~Cycle-Closing Edges}
Since we~already know that $W \le \operatorname{OPT}$ (see \eqref{eq:wopt}),
to~turn an~upper bound $W+O$ (see \eqref{eq:wo}) 
into an~approximation guarantee, 
we~need to~upper bound $O$~in~terms of~$\operatorname{OPT}$.
This is~exactly what many previous papers~do.
Namely, they prove an~upper bound of~the form
\begin{equation}
	O \le \operatorname{OPT} + \gamma W,
	\label{eq:o}
\end{equation}
for a~constant~$\gamma$.
Combined with~\eqref{eq:wopt} and~\eqref{eq:wo}, this inequality 
immediately implies that MGREEDY is~a~$(2+\gamma)$-approximation.
As~shown by~\cite[Section 5]{BJLTY91}, this implies that the approximation guarantee of~the greedy algorithm is~also at~most $(2+\gamma)$.

The MGREEDY algorithm simply concatenates the representative strings.
If,~instead, one applies an~$\alpha$-approximation algorithm of~MAX-ATSP in~the overlap graph
of~the representative strings and merges them in~the found order,
the resulting algorithm has approximation guarantee $(2+(1-\alpha)\gamma)$.
This was shown in~\cite{BJJ97,Mucha07,Mucha13}, a~proof can also be~found in~\cite[Theorem 3.1]{EMV22}.
The history of~improvements of~upper bounds on~$\gamma$ is~the following:
\cite{BJLTY91} proved the inequality with $\gamma=2$,
\cite{KS05} improved~it to~$3/2$, \cite{EMV22} pushed it~to~$1.425$,
and \cite{EMV23} improved it~further to~$1.396$.

We~prove the inequality~\eqref{eq:o} for $\gamma=1$ (see \Cref{thm:periodic-budget}).
This implies that the greedy algorithm is~$3$-approximate.
Plugging~in the best known approximation guarantee $\alpha=2/3$~\cite{KLSS03, PEZ12} for MAX-ATSP gives a~polynomial-time $7/3$-approximation for SCS.
Our proof of~a~stronger upper bound on~$\gamma$ does not come at~a~cost
of a~more involved analysis: in~fact, our proof is~much simpler.
We~achieve this by~considering global properties of~the overlap graph 
instead of~considering each cycle of~the optimal cycle cover separately.
Then, in~Theorem~\ref{prop:black-box-tightness}, we~prove that $\gamma=1$ is~the optimal endpoint of~this approach: for every $\gamma<1$, there exists 
an~input to~SCS where the inequality~\eqref{eq:o} fails.

\subsection{Properties of~Cycle Strings}
A~cycle with small total prefix may~be an~indicator of~the high repetitive structure of~the strings of~this cycle. For example, a~cycle through strings
\[{\tt abcdabcdabcd},\, {\tt cdabcdabcdab},\, {\tt dabcdabcdabc}\]
has total prefix~$2 + 1 + 1 = 4$, whereas its minimum edge overlap is~$10$. 
In particular, $O$ may be much larger than $W$.

Two strings are \emph{equivalent} if~one is~a~cyclic rotation of~the other.
A~string~$z$ is~\emph{primitive} if~there is~no~string~$v$ and an~integer $k\ge2$ such that $z=v^k$.
For a~cycle $C=(t_1,\dotsc,t_r,t_1)$ of~$\mathcal{C}$, define its \emph{cycle string} by
\[
	\tau(C)=\operatorname{pref}(t_1,t_2)\circ\cdots\circ\operatorname{pref}(t_r,t_1),
\] hence $|\tau(C)| = w(C)$.

\begin{lemma}[{\cite[Lemma 2.5]{Mucha13}}]\label{lem:primitive-cycle-words}
	Every cycle string~$\tau(C)$ is~primitive, and the cycle strings of~distinct cycles of~$\mathcal C$ are pairwise non-equivalent.
\end{lemma}

\begin{graybox}
	In~the running example, the three cycles~$C_1,C_2,C_3$ have the following cycle strings:
	\[
		\tau(C_1) =\texttt{abcd}, \qquad \tau(C_2) =\texttt{ba}, \qquad \tau(C_3) =\texttt{a}.
	\] 
	Then, $\sigma(C_1) = \tau(C_1)^2, \sigma(C_2) = \tau(C_2)^3$, and $\sigma(C_3) = \tau(C_3)^3$ (in general, $\sigma(C)$ is~not always a~power of~$\tau(C)$).
\end{graybox}

Fix an~arbitrary total order on~the alphabet.
A~\emph{Lyndon word}~\cite{Shirshov53,Lyndon54} is~a~string that is~strictly smaller than each of~its proper suffixes; equivalently, it~is~primitive and lexicographically smaller than all its nontrivial cyclic rotations.
Every equivalence class of~primitive strings has a~unique Lyndon word.
Moreover, a~Lyndon word~$z$ has no~period $p<|z|$: otherwise its suffix of~length $|z|-p$ would equal a~proper prefix of~$z$ and hence would be~smaller than~$z$, contrary to~the definition.

For a~string~$x$ and an~integer $m\ge0$, let $\operatorname{Sub}_m(x)$ be~the set of~length-$m$ substrings of~$x$.
For a~string~$z$ and an~integer $m\ge0$, let $\operatorname{CSub}_m(z)$ be~the set of~length-$m$ substrings of~$z^\infty$.
For every cycle~$C\in \mathcal{C}$, let~$\tau'(C)$ be~the unique Lyndon word equivalent to~$\tau(C)$.

\begin{lemma}[{\cite[Claim 2]{BJLTY91}}]\label{lem:cyclic-substrings}
	For every cycle~$C\in \mathcal{C}$,
	\[\!\operatorname{CSub}_{o(C)+1}(\tau'(C))=\operatorname{CSub}_{o(C)+1}(\tau(C))\subseteq\operatorname{Sub}_{o(C)+1}(\operatorname{OPT}).\]
\end{lemma}

\begin{graybox}
	Consider the optimal superstring
	\(\operatorname{OPT}=\texttt{aaabababcdabcd}.\)
	Using the alphabet order $\texttt{a}<\texttt{b}<\texttt{c}<\texttt{d}$, the Lyndon words corresponding to~the three cycle strings are
	\[\tau'(C_1)=\texttt{abcd},\qquad \tau'(C_2)=\texttt{ab},\qquad \tau'(C_3)=\texttt{a}.\]
	The minimum overlaps of~the three cycles are $o(C_1)=o(C_2)=4$ and~$o(C_3)=2$.
	Then,
	\begin{align*}
		\operatorname{CSub}_{5}(\tau'(C_1))&= \{\texttt{abcda},\texttt{bcdab},\texttt{cdabc},\texttt{dabcd}\}, \\
		\operatorname{CSub}_{5}(\tau'(C_2))&=\{\texttt{ababa},\texttt{babab}\}, \\
		\operatorname{CSub}_{3}(\tau'(C_3))&=\{\texttt{aaa}\}.
	\end{align*}
	Thus, every length-$(o(C)+1)$ substring of~$\tau(C)^\infty$ occurs as~a~substring of~$\operatorname{OPT}$.
\end{graybox}

\section{Optimal Bound on~the Cycle-Cover Inequality}\label{sec:periodic}
In~this section, we~give a~proof of~our main technical contribution: the smallest value of~$\gamma$ where the inequality~\eqref{eq:o} holds is~equal to~1.

\subsection{Upper Bound}

\begin{theorem}
	\label{thm:periodic-budget}
	\[O \le \operatorname{OPT}+W.\]
\end{theorem}

\Cref{lem:cyclic-substrings} shows that, for every cycle~$C$, all length-$(o(C)+1)$ substrings of~$\tau(C)^\infty$ occur in~$\operatorname{OPT}$.
The lemma below bounds the number of~cycles that can use the same superstring.

\begin{lemma}\label{lem:substring-budget}
	Let~$u$ be~a~string, $m\in\{1,\ldots,|u|\}$, and let $z_1,\ldots,z_a$ be~distinct Lyndon words for some $a\ge0$.
	Suppose that $|z_j|\le m$ and $\operatorname{CSub}_m(z_j)\subseteq\operatorname{Sub}_m(u)$ for every~$j$.
	Then
	\[a\le|\!\operatorname{Sub}_m(u)|-|\!\operatorname{Sub}_{m-1}(u)|+1.\]
\end{lemma}

\begin{graybox}
	Consider the difference between $\operatorname{Sub}_{m}$ and $\operatorname{Sub}_{m - 1}$.
	Every distinct length-$(m-1)$ substring of~$u$, except possibly the one occurring only at~the end of~$u$, can~be extended to~the right by~one symbol.
	Therefore
	\[\!|\operatorname{Sub}_m(u)|\ge|\operatorname{Sub}_{m-1}(u)|-1.\]
	\Cref{lem:substring-budget} shows that each Lyndon word whose complete set of~length-$m$ periodic substrings occurs in~$u$ forces one additional length-$m$ substring.
	To~give an~example, let $u=\texttt{aababbabaab}$ and~$m=4$.
	Then
	\begin{align*}
		\operatorname{Sub}_3(u)&=\{\texttt{aab},\texttt{aba},\texttt{abb},\texttt{baa},\texttt{bab},\texttt{bba}\},\\
		\operatorname{Sub}_4(u)&=\{\texttt{aaba},\texttt{abaa},\texttt{abab},\texttt{abba}, \texttt{baab},\texttt{baba},\texttt{babb},\texttt{bbab}\}.
	\end{align*}
	By~the extension argument above, $|\operatorname{Sub}_4(u)| \ge |\operatorname{Sub}_3(u)|-1=5$.

	However, consider the Lyndon words~$z_1=\texttt{aab}$, $z_2=\texttt{ab}$, and~$z_3=\texttt{abb}$.
	Their length-$4$ cyclic substrings are
	\begin{align*}
		\operatorname{CSub}_4(z_1)&=\{\texttt{aaba},\texttt{abaa},\texttt{baab}\},\\
		\operatorname{CSub}_4(z_2)&=\{\texttt{abab},\texttt{baba}\},\\
		\operatorname{CSub}_4(z_3)&=\{\texttt{abba},\texttt{babb},\texttt{bbab}\}.
	\end{align*}
	They are contained in~$\operatorname{Sub}_4(u)$, so~\Cref{lem:substring-budget} implies that these three Lyndon words account for three additional length-$4$ substrings, which shows that
	\[
	|\operatorname{Sub}_{4}(u)| \ge a + |\operatorname{Sub}_3(u)| - 1 = 8.
	\] 
\end{graybox}

\begin{proof}
	Let~$G_m$ be~the following undirected multigraph: 
	its vertices are~the strings in~$\operatorname{Sub}_{m-1}(u)$, and each string $\operatorname{Sub}_m(u)$ is~an~edge joining its length-$(m-1)$ prefix and its length-$(m-1)$ suffix (it~is an~undirected version
	of~the order-$(m-1)$ Rauzy graph~\cite{Rauzy82} of~$u$, which~is in~turn a~subgraph
	of~the de~Bruijn graph~\cite{DeBrujin46}).
	Thus, $G_m$ has $|\!\operatorname{Sub}_{m-1}(u)|$ vertices and $|\!\operatorname{Sub}_m(u)|$ edges.
	It~is~connected because the consecutive length-$(m-1)$ substrings of~$u$ form a~walk that visits every vertex.
	
	For every~$j$, let~$b_j$ be~the length-$m$ prefix of~$z_j^\infty$.
	Since $m\ge|z_j|$ and~$z_j$ is~a~Lyndon word, $b_j$ is~the unique lexicographically smallest string in~$\operatorname{CSub}_m(z_j)$.
	The strings~$b_1,\ldots,b_a$ are pairwise distinct.
	Indeed, suppose $b_i=b_j$.
	If~$|z_i|=|z_j|$, since $m\ge |z_i|$, the length-$|z_i|$ prefixes of~$b_i$ and~$b_j$ are~$z_i$ and~$z_j$, respectively, so~$z_i=z_j$.
	If, say, $|z_i|<|z_j|$, then $z_j$ is~a~prefix of~$z_i^\infty$ and hence has period $|z_i|<|z_j|$, which is~impossible for a~Lyndon word.
	After relabeling, assume $b_1<b_2<\cdots<b_a$ and delete the edges~$b_1,b_2,\ldots,b_a$ from~$G_m$ in~that order.
	We~show that after each deletion the graph remains connected.

	Consider the deletion of~the edge~$b_j$, and, for $0\le i<|z_j|$, let~$c_i$ and~$u_i$ be~the length-$m$ and length-$(m-1)$ substrings of~$z_j^\infty$ beginning at~offset~$i$, respectively.
	Then $c_0=b_j$, and each~$c_i$ is~an~edge of~$G_m$ connecting $u_i$ and~$u_{i+1}$.
	If~$|z_j|>1$, deleting the edge~$c_0 = b_j$ keeps the graph connected, since the alternate walk $c_1,c_2,\ldots,c_{|z_j|-1}$ is~still present: each $c_i$ with $1\le i<|z_j|$ is~larger than~$b_j$, whereas every previously deleted edge is~smaller than~$b_j$.
	If~$|z_j|=1$, then $b_j$ is~a~loop, so~its deletion is~harmless.

	\begin{graybox}
		For the same string $u=\texttt{aababbabaab}$ and~$m=4$, the graph~$G_4$ is~shown below.

		
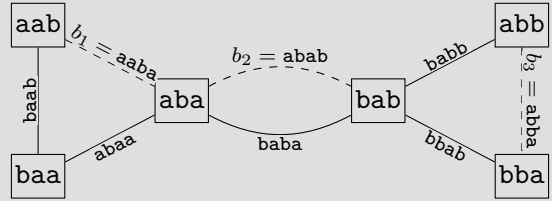
\begin{center}
			\begin{tikzpicture}
				\foreach \s/\x/\y in {aab/-32/10, aba/-13/0, baa/-32/-10, bab/13/0, abb/32/10, bba/32/-10}
					\node[v] (\s) at (\x mm,\y mm) {\strut \tt \s};

				\draw[dashed] (aab) -- node[l,pos=.5,above] {\scriptsize $b_1=\mathtt{aaba}$} (aba);
				\draw (aba) -- node[l,pos=.5,below] {\scriptsize $\mathtt{abaa}$} (baa);
				\draw (baa) -- node[l,pos=.5,above] {\scriptsize $\mathtt{baab}$} (aab);
				\draw[dashed] (aba) to[bend left=28] node[l,pos=.5,above] {\scriptsize $b_2=\mathtt{abab}$} (bab);
				\draw (aba) to[bend right=28] node[l,pos=.5,below] {\scriptsize $\mathtt{baba}$} (bab);
				\draw (bab) -- node[l,pos=.5,above] {\scriptsize $\mathtt{babb}$} (abb);
				\draw[dashed] (abb) -- node[l,pos=.5,above] {\scriptsize $b_3=\mathtt{abba}$} (bba);
				\draw (bba) -- node[l,pos=.5,below] {\scriptsize $\mathtt{bbab}$} (bab);
			\end{tikzpicture}
		\end{center}
	\end{graybox}

	After all $a$ deletions, the remaining graph is~connected and has $|\!\operatorname{Sub}_m(u)|-a$ edges; consequently, $|\!\operatorname{Sub}_m(u)|-a\ge|\!\operatorname{Sub}_{m-1}(u)|-1$.
\end{proof}

\begin{proof}[Proof of~\Cref{thm:periodic-budget}]
	Enumerate the cycles~$C\in \mathcal{C}$ satisfying $o(C)+1\ge w(C)$ as~$C_1,\ldots,C_h$.
	For every~index $j$, put $z_j=\tau'(C_j)$ and $\ell_j=o(C_j)+1$.
	The Lyndon words~$z_1,\ldots,z_h$ are distinct by~\Cref{lem:primitive-cycle-words}, and $|z_j|=|\tau(C_j)|=w(C_j)\le\ell_j$ for every~$j$.
	Moreover, \Cref{lem:cyclic-substrings} gives $\operatorname{CSub}_{\ell_j}(z_j)\subseteq\operatorname{Sub}_{\ell_j}(\operatorname{OPT})$ for every~$j$.
	For $1\le m\le \operatorname{OPT}$, let~$a_m$ be~the number of~indices~$j$ satisfying $|z_j|\le m\le\ell_j$.
	\Cref{lem:substring-budget} gives
	\[a_m\le|\!\operatorname{Sub}_m(\operatorname{OPT})|-|\!\operatorname{Sub}_{m-1}(\operatorname{OPT})|+1.\]
	Each index~$j$ is~counted by~$a_m$ for the values $m=|z_j|,\ldots,\ell_j$, hence exactly $\ell_j-|z_j|+1=o(C_j)-w(C_j)+2$ times.
	Since $|\!\operatorname{Sub}_0(\operatorname{OPT})|=|\!\operatorname{Sub}_{\operatorname{OPT}}(\operatorname{OPT})|=1$, summing the inequalities over~$m$ gives 
	\begin{align*}
		\sum_{\substack{C\in \mathcal{C}\\o(C)+1\ge w(C)}}\bigl(o(C)-w(C)+2\bigr)&=\sum_{j=1}^{h}\bigl(\ell_j-|z_j|+1\bigr)\\
		&=\sum_{m=1}^{\operatorname{OPT}}a_m\\
		&\le\sum_{m=1}^{\operatorname{OPT}}\bigl(|\!\operatorname{Sub}_m(\operatorname{OPT})|-|\!\operatorname{Sub}_{m-1}(\operatorname{OPT})|+1\bigr)\\
		&=\operatorname{OPT}.
	\end{align*}
	For each cycle in~the first sum, $o(C)-w(C)\le o(C)-w(C)+2$, whereas every omitted cycle satisfies $o(C)-w(C)<0$.
	Therefore, 
	\[
	O - W \le \operatorname{OPT}. \qedhere
	\] 
\end{proof}

\subsection{Lower Bound}

In~this section, we~prove that the coefficient of~$W$ in~\Cref{thm:periodic-budget} cannot be~decreased further.

\begin{theorem}\label{prop:black-box-tightness}
	For every $\gamma<1$, 
	there is~a~substring-free set~$\mathcal{S}$ of~binary strings 
	such that
	\[O>\operatorname{OPT}+\gamma W.\]
\end{theorem}
\begin{proof}
	For $p\ge2$, put $s_p={\tt 0}^{p-1}{\tt 1}{\tt 0}^{p-1}{\tt 1}{\tt 0}^{p-2}$.
	\begin{graybox}
		\centering
		\(\begin{aligned}
			s_2 &= \texttt{0101}\\
			s_3 &= \texttt{0010010}\\
			s_4 &= \texttt{0001000100}\\
			s_5 &= \texttt{0000100001000}
		\end{aligned}\)
	\end{graybox}
	Clearly, 
	$s_p$ starts and ends with ${\tt 0}^{p-1}{\tt 1}{\tt 0}^{p-2}$, so $|\!\operatorname{ov}(s_p, s_p)| \ge 2p-2$. Also, for $p \neq q$,
	$s_p$ is not a~substring of~$s_q$ 
	($s_p$ has exactly two $\tt 1$'s, whose positions differ by~$p$)
	and
	$|\!\operatorname{ov}(s_p, s_q)| \le 2p - 2$ (a~longer overlap would cover both $\tt 1$'s in~$s_p$).
	
	For $m\ge2$, let $\mathcal{S}_m=\{s_2,s_3,\dotsc,s_m\}$.
	By~the discussion above, $\mathcal{S}_m$ is~substring-free
	and its cycle cover consisting of~all self-loops has~minimum total prefix.
	Then,
	\begin{align*}
		W&=\sum_{p=2}^{m} (3 p - 2 - (2p - 2)) =\sum_{p = 2}^{m} p = m (m + 1) / 2 - 1, \\
		O&=\sum_{p=2}^{m}(2p-2)=2W-2(m-1).
	\end{align*}
	Since $|\!\operatorname{ov}(s_p,s_{p-1})|=2p-3$,
	\[
	\operatorname{OPT} \le
	|\langle s_m,s_{m-1},\dotsc,s_2\rangle| = |s_2|+\sum_{p=3}^{m}\bigl(|s_p|-(2p-3)\bigr)=4+\sum_{p=3}^{m}(p+1)=W+m.\]
	Consequently,
	\[O-\operatorname{OPT}\ge2W-2(m-1)-(W+m)=W-3m+2.\]
	Since $W=\Theta(m^2)$, for every fixed $\gamma<1$ and all sufficiently large~$m$, $O-\operatorname{OPT}>\gamma W$.
\end{proof}

\section*{AI disclosure}
We~used ChatGPT 5.6-Sol Pro to~explore different approaches to~prove the greedy conjecture.
While we~were trying to~formulate a~self-reduction for SCS, the agent suggested a~proof of~the inequality in~\Cref{thm:periodic-budget}, it~was not directly related to~the discussed question, but the agent had found a~simpler route to~the statement we~tried to prove.
The initial proof was complicated, but working together we~simplified it~significantly.
Then we~used the agent to~do mechanical work to~construct the example in~\Cref{prop:black-box-tightness}.
We~also used Codex 5.6-Sol Ultra to~find all typos in~the paper.
We~have written all~the arguments by~ourselves and assume full responsibility for~the manuscript.

\bibliographystyle{alpha}
\bibliography{refs}
\end{document}